\documentclass[journal]{IEEEtran}

\usepackage{amsthm,amsmath,graphicx,algorithm,algorithmic,threeparttable,epstopdf,cite}

\newtheorem{assumption}{Assumption}
\newtheorem{theorem}{Theorem}
\newtheorem{remark}{Remark}
\newtheorem{corollary}{Corollary}

\begin{document}

\title{Integration of Nonlinear Disturbance \\ Observer within Proxy-based Sliding Mode \\ Control for Pneumatic Muscle Actuators}

\author{Yu~Cao, 
        Jian~Huang$^*$, 
        Dongrui~Wu$^*$, 
        Mengshi~Zhang, 
        Caihua~Xiong, 
        Zhijun~Li,
%
\thanks{Y. Cao, J. Huang, D. Wu and M. Zhang are with the Key Laboratory for Image Processing and Intelligent Control of Education Ministry of China, School of Artificial Intelligence and Automation, Huazhong University of Science and Technology, Wuhan, China (e-mail: cao\_yu@mail.hust.edu.cn; huang\_jan@mail.hust.edu.cn; drwu@hust.edu.cn; dream\_poem@hust.edu.cn). }
\thanks{C.-H. Xiong is with the School of Mechanical Science Engineering and the State Key Laboratory of Digital Manufacturing Equipment and Technology, Huazhong University of Science and Technology, Wuhan 430074, China (e-mail: chxiong@hust.edu.cn). }
\thanks{Z. Li is with the College of Automation Science and Engineering, South China University of Technology, Guangzhou 510641, China (email: zjli@ieee.org).}
}

\markboth{XXX,~Vol.~xx, No.~x, January~2019}%
{Shell \MakeLowercase{\textit{et al.}}: Bare Demo of IEEEtran.cls for IEEE Journals}

\maketitle

\begin{abstract}
This paper presents an integration of nonlinear disturbance observer within proxy-based sliding mode control (IDO-PSMC) approach for Pneumatic Muscle Actuators (PMAs). Due to the nonlinearities, uncertainties, hysteresis, and time-varying characteristics of the PMA, the model parameters are difficult to be identified accurately, which results in unmeasurable uncertainties and disturbances of the system. To solve this problem, a novel design of proxy-based sliding mode controller (PSMC) combined with a nonlinear disturbance observer (DO) is used for the tracking control of the PMA. Our approach combines both the merits of the PSMC and the DO so that it is effective in both reducing the ``chattering" phenomenon and improving the system robustness. A constrained Firefly Algorithm is used to search for the optimal control parameters. Based on the Lyapunov theorem, the states of the PMA are shown to be globally uniformly ultimately bounded. Extensive experiments were conducted to verify the superior performance of our approach, in multiple tracking scenarios.
\end{abstract}


\IEEEpeerreviewmaketitle

\section{Introduction}

\IEEEPARstart{T}{he} Pneumatic Muscle Actuator (PMA) has been widely used in a variety of fields, due to its attractive characteristics, i.e., high power/weight ratio, no mechanical parts, low cost, etc \cite{Dzahir:2014, Andrikopoulos:2011}. Its driving force is converted from the the air pressure of the inner bladder, which has the features of nonlinearity, hysteresis, and time-varying parameters \cite{Caldwell:1995}, making its modeling and control very challenging. Different control strategies have been proposed for the PMA, including PID-based control \cite{Andrikopoulos:2014,Zhang:2017}, sliding mode control (SMC) \cite{Cao:2018}, nonlinear model predictive control \cite{Huang:2018, Huang:2016}, fuzzy control \cite{Xie:2011}, adaptive control \cite{Zhu:2017}, etc. Most of them are model-based, which require an accurate mathematic model of the PMA. Unfortunately, such an accurate model is very difficult to obtain in practice. Furthermore, different applications require the PMA to accurately track various reference trajectories with different loads and motion frequencies. Thus, there is a strong demand for robust PMA control strategies.

SMC is a well-known model-based approach to deal with uncertain systems, due to its ability to handle variations and external disturbances. However, it could lead to a ``chattering" phenomenon that may cause serious damages to the actuator. A remedy is proxy-based sliding mode control (PSMC) \cite{Kikuuwe:2010}, which does not need an accurate system model, and hence the complicated system identification process can be avoided. In addition, through the ``proxy", PSMC can make the system compliant to external disturbances and reduce the ``chattering" phenomenon significantly. Thus, it has been successfully used in different applications \cite{Gu:2015, Damme:2009, Chen:2016}. However, the stability analysis of PSMC depends on a strong conjecture (see Conjecture~1 in \cite{Kikuuwe:2010}), which may not always be satisfied in practice.

Nonlinear Disturbance Observer (DO) based control is a common strategy for improving the control performance. Its basic idea is to estimate the disturbances/uncertainties from measurable variables before a control action is taken. Consequently, the influence of the disturbances/uncertainties can be suppressed, and the system becomes more robust \cite{Ohishi:1987, Han:2009}. Multiple DO-based control strategies have been proposed to compensate the influence of disturbances/uncertainties \cite{Ginoya:2014, Huang:2015, Huang:2018-1, Chen:2017}.
However, to our best knowledge, there has not been any research on DO-based PSMC. This may be due to two challenges. First, the PSMC is a model-free control strategy, whereas a typical DO-based controller requires a mathematical model of the controlled object. Therefore, the integration of PSMC and DO is not straightforward. Second, a more rigorous analysis is needed to guarantee the stability of DO-based PSMC, which should not be based on the strong conjecture in \cite{Kikuuwe:2010}. Additionally, in practice the control performance is largely determined by the optimality of the control parameters, which are not easy to be obtained, especially when there are many such parameters. Traditionally, the control parameters are empirically tuned by the controller designer, which takes time and experience. An automated optimization approach is highly desirable.

This paper proposed an integration of nonlinear disturbance observer within proxy-based sliding mode control (IDO-PSMC), automatically tuned by a constrained Firefly Algorithm (FA), which can eliminate the ``chattering" phenomenon and increase the robustness of the system. Our main contributions are: 1) the proposed control strategy called IDO-PSMC for PMA tracking tasks, with automatic parameter optimization; 2) theoretical analysis on the stability of the closed-loop system and the effect of the proxy mass; 3) real-world experiments for validating the effectiveness and robustness of the proposed control strategy with various reference trajectories.

The rest of this paper is organized as follows. Section~II introduces the three-element model of the PMA with disturbance. Section~III proposed the IDO-PSMC. Section~IV analyzes its stability. Sections~V presents real-world experiments to demonstrate the effectiveness and robustness of the IDO-PSMC. Finally, Section~VI draws conclusions.

\section{The Three-Element Model of the PMA}

The generalized three-element model of the PMA is a parallel connection of a contractile element, a spring element, and a damping element \cite{Reynolds:2003}, as shown in Fig.~1. The contractile length varies with the air pressure of inner bladder. The dynamics of the PMA can be expressed as:
\begin{align} \label{2.1}
\left\{ {\begin{array}{*{20}{c}}
{m\ddot x + b(P)\dot x + k(P)x = f(P) - mg}\\
{{b_i}(P) = {b_{i0}} + {b_{i1}}P} \quad (inflation)\\
{{b_d}(P) = {b_{d0}} + {b_{d1}}P} \quad (deflation)\\
{\begin{array}{*{20}{c}}
{k(P) = {k_0} + {k_1}P}\\
{f(P) = {f_0} + {f_1}P}
\end{array}}
\end{array}} \right.
\end{align}
where $m$, $x$, $P$ are the mass of load, the contractile length of PMA, and the air pressure, respectively. $b(P)$, $f(P)$, $k(P)$ are the damping coefficient, the contractile force, and the spring coefficient, respectively. The damping element is a highly nonlinear function of the air pressure $P$, which is difficult to identify in practice. So, a piecewise linear function is used to approximate the nonlinear one and to represent the hysteresis in the dynamic model.

\begin{figure} [t] \centering
\includegraphics[width=3.4in,height=2.4in]{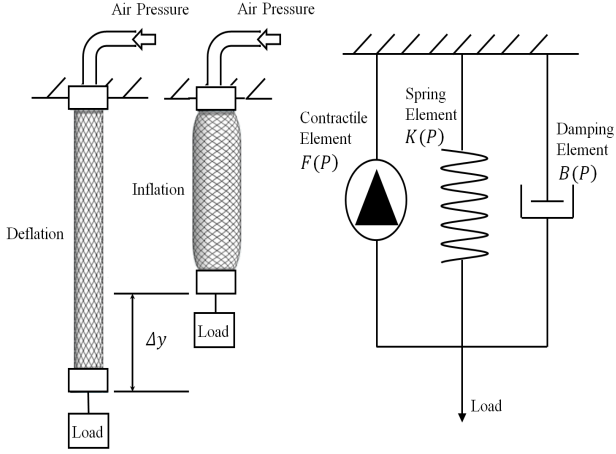}
\caption{The PMA (a) and its three-element model (b).}
\label{PMA_MODEL_FIG}
\end{figure}

Clearly, the modeling and control of the PMA involves lots of uncertainties, including modeling error, friction, inaccurate parameters, changing loads, etc. Let $\tau (t)$ denote the sum of these uncertainties. Then, the dynamic model of the PMA can be rewritten as:
\begin{align}  \label{2.2}
\left\{ \begin{array}{l}
\ddot x = f(x,\dot x) + b(x,\dot x)u + \tau (t)\\
f(x,\dot x) = \frac{1}{m}({f_0} - mg - {b_0}\dot x - {k_0}x)\\
b(x,\dot x) = \frac{1}{m}({f_1} - {b_1}\dot x - {k_1}x)
\end{array} \right.
\end{align}
where $u$ is the air pressure, and $f(x,\dot x)$ and $b(x,\dot x)$ are nonlinear functions.

In tracking applications, (\ref{2.2}) can be viewed as a single-input single-output (SISO) system, in which the input $u$ is the input air pressure, and the output $x$ is the displacement of the PMA. It is also a typical second-order nonlinear model with disturbances. $b_0$ ($b_1$) will be selected as $b_{i0}$ ($b_{i1}$) or $b_{d0}$ ($b_{d1}$) according to the states of the PMA aeration, and $k_0$ ($k_1$) as $k_{01}$ ($k_{02}$) or $k_{11}$ ($k_{12}$) according to the inner pressure of the PMA.

\section{Integration of Nonlinear Disturbance Observer within Proxy-based Sliding Mode Control for the PMA}

This section introduces our proposed controller for the PMA.

\subsection{The Nonlinear Disturbance Observer (DO)}

The model disturbances/uncertainties are inevitable and usually cannot be directly measured in real-world applications, which deteriorate the control performance of the PMA. However, a nonlinear DO can be used to estimate them from measurable variables of the PMA. A second-order DO, first proposed in \cite{Ginoya:2014}, was used in our study. Also, the results can also be easily extended to controllers with higher-order DOs.

The following assumption was used in our design:
\begin{assumption}
The disturbance $\tau (t)$ is continuous and satisfies
\begin{align}
\left|\frac{{{d^i}\tau (t)}}{{d{t^i}}}\right| \le \varepsilon,  \qquad  for \quad i = 0, 1, 2\label{3.1}
\end{align}
where $\varepsilon$ is a positive number.
\end{assumption}

Then, the nonlinear DO is represented as:
\begin{align}
\hat \tau  &= {p_1} + {l_1}\dot x\label{3.2} \\
{\dot p_1} &=  - {l_1}\left( {f(x,\dot x) + b(x,\dot x)u + \hat \tau } \right) + \hat {\dot \tau }\label{3.3}\\
\hat {\dot \tau} & = {p_2} + {l_2}\dot x\label{3.4}\\
{{\dot p}_2}& =  - {l_2}\left( {f(x,\dot x) + b(x,\dot x)u + \hat {\tau} } \right) \label{3.5}
\end{align}
where $\hat \tau$ and $\hat {\dot \tau }$ are the estimates of $\tau$ and $\dot \tau$, respectively. $p_1$ and $p_2$ are auxiliary variables. $l_1$ and $l_2$ are positive constants.

The estimation errors is:
\begin{align} \label{3.6}
\tilde \tau  &= \tau  - \hat \tau \\
\tilde {\dot \tau}  &= \dot \tau  - \hat {\dot \tau}
\end{align}
where $\tilde \tau$ is the estimation error of $\tau$, and $\tilde {\dot \tau}$ the estimation error of $\dot \tau$.

Assumption 1 implies that the disturbance $\tau(t)$ is continuous, and ${{d^i}\tau (t)}/{d{t^i}}(i = 0, 1, 2)$ is bounded. By using the similar technique in \cite{Ginoya:2014}, the dynamics of the error of the nonlinear DO is:
\begin{align} \label{3.12}
{\bf{\dot {\tilde e}}} = {\bf{A}}_1 {\bf{\tilde e}} + {\bf{B}}_1\ddot \tau
\end{align}
where
\[{\bf{\tilde e}} = \left[ {\begin{array}{*{20}{c}}
{\tilde \tau }\\
{\tilde {\dot \tau} }
\end{array}} \right],{\bf{A}}_1 = \left[ {\begin{array}{*{20}{c}}
{ - {l_1}}&1\\
{ - {l_2}}&0
\end{array}} \right],{\bf{B}}_1 = \left[ {\begin{array}{*{20}{c}}
0\\
1
\end{array}} \right]\]
The eigenvalues of ${\bf{A}}_1$ can be placed arbitrarily. We select ${l_1}$ and ${l_2}$ to make the real part of these eigenvalues negative.

There exits a positive definite matrix ${\bf{P}}_1$ satisfying
\begin{align} \label{3.1.13}
{\bf{A}}_1^T{{\bf{P}}_1} + {{\bf{P}}_1}{{\bf{A}}_1} =  - {{\bf{Q}}_1}
\end{align}
for any given positive definite matrix ${{\bf{Q}}_1}$.

Choosing a Lyapunov function
\begin{align} \label{3.1.14}
V({\bf{\tilde e}}) = {{\bf{\tilde e}}^T}{{\bf{P}}_1}{\bf{\tilde e}}
\end{align}
and differentiating $V({\bf{\tilde e}})$, we have
\begin{align} \label{3.1.15}
\dot V({\bf{\tilde e}}) &= {{\dot{\bf{ \tilde e}}^T}}{{\bf{P}}_1}{\bf{\tilde e}} + {{{\bf{\tilde e}}}^T}{{\bf{P}}_1}\dot {\bf{\tilde e}}\\
 &=  - {{{\bf{\tilde e}}}^T}{{\bf{Q}}_1}{\bf{\tilde e}} + 2{{{\bf{\tilde e}}}^T}{{\bf{P}}_1}{{\bf{B}}_1}\ddot \tau \\
& \le  - {\lambda _{\min }}({{\bf{Q}}_1}){\left\| {{\bf{\tilde e}}} \right\|^2} + 2\left\| {{{\bf{P}}_1}{{\bf{B}}_1}} \right\|\left\| {{\bf{\tilde e}}} \right\|\varepsilon \\
& \le  - \left\| {{\bf{\tilde e}}} \right\|({\lambda _{\min }}({{\bf{Q}}_1})\left\| {{\bf{\tilde e}}} \right\| - 2\left\| {{{\bf{P}}_1}{{\bf{B}}_1}} \right\|\varepsilon )
\end{align}
where $|| \cdot ||$ denotes the 1-norm. Therefore, the norm of the estimation error is bounded by:
\begin{align} \label{3.1.16}
\left\| {{\bf{\tilde e}}} \right\| \le {\lambda _{\rm{1}}}
\end{align}
where
\[{\lambda _{\rm{1}}}{\rm{ = }}\frac{{2\left\| {{{\bf{P}}_1}{{\bf{B}}_1}} \right\|\varepsilon }}{{{\lambda _{\min }}({{\bf{Q}}_1})}}\]

Hence, if Assumption 1 holds, then the disturbance estimation error is uniformly ultimately bounded.

\subsection{Integration of Nonlinear Disturbance Observer within Proxy-based Sliding Mode Control}

In the PSMC, an imaginary object called ``proxy", assumed to be connected to the physical actuator, is presented. Before introducing the IDO-PSMC, we define the following sliding manifolds:
\begin{align}
{S_q} &= {{\dot x}_d} - \dot x + c_1({x_d} - x) + c_2\displaystyle{\int {\left( {{x_d} - x} \right)dt}}\label{3.13}\\
{S_p} &= {{\dot x}_d} - {\dot x}_p + c_1({x_d} - {x_p}) + c_2\displaystyle{\int {\left( {{x_d} - {x_p}} \label{3.14} \right)dt}}
\end{align}
where $c_1$ and $c_2$ are positive constants, $x_d$ the desired trajectory, and $x_p$ and $x$ the proxy position and the PMA's displacement, respectively. The sliding manifolds directly reflect the tracking states of the proxy and the PMA. The diagram of the IDO-PSMC is shown in Fig.~\ref{IDO-PSMC structure}.

\begin{figure} [htpb]\centering
\includegraphics[width=3.2in,height=1.78in]{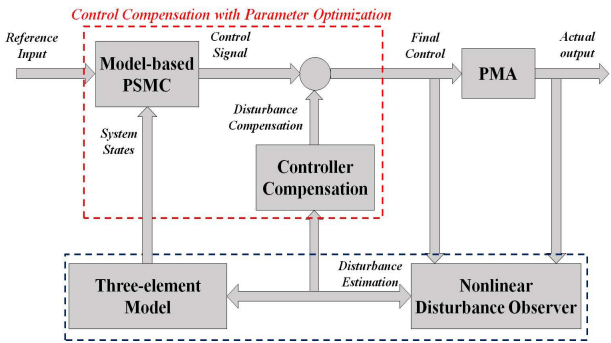}
\caption{The diagram of the proposed IDO-PSMC.}
\label{IDO-PSMC structure}
\end{figure}

A PID-type virtual coupling is adopted to drive the states of the PMA ${\bf{X}} = {[\int x dt,x,\dot x]^T}$ to the states of the proxy ${{\bf{X}}_p} = {[\int {x_p} dt,{x_p},{\dot x_p}]^T}$. By taking the DO into consideration, we design the following relation based on the idea of the SMC:
\begin{align} \label{3.15}
\begin{array}{l}
{{\dot S}_q} + {K_p}({x_p} - x) + {K_i}\displaystyle{\int {\left( {{x_p} - x} \right)dt}} \\
\qquad \qquad + {K_d}({{\dot x}_p} - \dot x) + \tilde \tau  - \hat {\dot \tau}  = 0
\end{array}
\end{align}
where $K_p$, $K_i$ and $K_d$ are positive constants.

\begin{remark}
Once the sliding mode manifold $S_q$ is defined, the controller can be designed using ${\dot S_q} =  - k\cdot {\mathop{\rm sgn}} ({S_q})$, which may cause severe chattering. Hence, we replace $- k{\mathop{\rm sgn}} ({S_q})$ by a PID controller to establish a connection between the controlled object and the proxy, as shown in Fig.~\ref{RPSMC-principle}. Note that (\ref{3.15}) can also be rewritten as:
\begin{align} \label{3.16}
{\bf{\dot X}} = \left[ {\begin{array}{*{20}{c}}
   0 & 1 & 0  \\
   0 & 0 & 1  \\
   0 & { - {c_2}} & { - {c_1}}  \\
\end{array}} \right]{\bf{X}} + \left[ {\begin{array}{*{20}{c}}
   0  \\
   0  \\
   1  \\
\end{array}} \right]{u_l} + \left[ {\begin{array}{*{20}{c}}
   0  \\
   0  \\
   1  \\
\end{array}} \right]\rho
\end{align}
where ${u_l} ={K_p}({x_p} - x) + {K_i}\displaystyle{\int {\left( {{x_p} - x} \right)dt}}  + {K_d}({{\dot x}_p} - \dot x)$, and ${\bf{\rho }} = {{\ddot x}_d} + c_1{{\dot x}_d} + c_2 {x_d} + \tilde \tau  - \hat {\dot \tau}$.

It is clear that (\ref{3.16}) can be regraded as the local relation between the controlled object and the proxy. This is a linear system with PID control, where ${\bf{X}}$ are the system's states, ${\bf{X}}_p$ the desired states, and ${{\ddot x}_d}$, ${\dot {x}}_d$ and $x_d$ varying parameters unrelated to the system's states. This PID controller drives the PMA's states ${\bf{X}}$ to the proxy's states ${\bf{X}}_p$, if the controller parameters are tuned properly.
\end{remark}

\begin{figure}[htpb] \centering
\includegraphics[width=3.2in,height=0.8in]{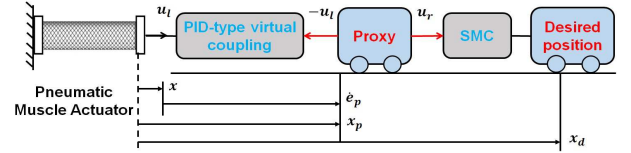}
\caption{The principle of proxy-based sliding mode control.}
\label{RPSMC-principle}
\end{figure}

Meanwhile, the DO can be regarded as compensation to the controller. Hence, the control signal fed into the PMA can be obtained according to (\ref{2.2}), (\ref{3.6}), and (\ref{3.15}):
\begin{align} \label{3.17}
\begin{array}{*{20}{l}}
{u = \frac{1}{{b(x,\dot x)}}[{{\ddot x}_d} + {c_1}({{\dot x}_d} - \dot x) + {c_2}({x_d} - x) - f(x,\dot x)}  \\
+ {K_p}({x_p} - x) + {K_i}\displaystyle{\int {\left( {{x_p} - x} \right)dt}}  + {K_d}({{\dot x}_p} - \dot x) \\
 - \hat \tau  - \hat {\dot \tau} ]  \\
\end{array}
\end{align}

An equivalent sliding mode controller is applied to generate the control signal $u_r$ between the reference and the proxy, i.e.,
\begin{align} \label{3.20}
{u_r} = \Gamma\cdot {\mathop{\rm sgn}} ({S_p})
\end{align}
where ${\rm{{\Gamma} > 0}}$, and $sgn(S_p)$ is the signum function.

Let $m_p>0$ be the so-called proxy mass. Then,
\begin{align} \label{3.21}
{m_p}{\dot S_p} =  - {u_r} + {u_l}
\end{align}
i.e., the proxy is simultaneously influenced by the virtual coupling and the equivalent sliding mode controller.

Combining (\ref{3.14}), (\ref{3.20}), and (\ref{3.21}), we have
\begin{align} \label{3.26}
{{{\ddot x}_p}{\rm{ = }}\frac{1}{{{m_p}}}(\Gamma {\rm{sgn}}({S_p}) - {K_p}({x_p} - x) - {K_i}\displaystyle{\int {\left( {{x_p} - x} \right)dt}} }  \nonumber\\
- {K_d}({{\dot x}_p} - \dot x)) + {{\ddot x}_d} + {c_1}({{\dot x}_d} - {{\dot x}_p}) + {c_2}({x_d} - {x_p})
\end{align}
The control signal can then be computed from (\ref{3.14}), (\ref{3.17}) and (\ref{3.26}).

\subsection{Parameters Optimization}

There are a number of control parameters to be determined during the controller design, including $\Gamma$, $c_1$, $c_2$, ${K_p}$, ${K_i}$, ${K_d}$, $l_1$, and $l_2$. A constrained FA is used to tune them.

The FA \cite{Yang:2009} is based on three assumptions:
\begin{enumerate}
\item All fireflies are unisex so that one firefly is attracted to others regardless of sex.

\item The attractiveness of a firefly is mainly determined by its brightness, which means that a less bright firefly will move towards a brighter one.

\item The brightness of a firefly is determined by the objective function.
\end{enumerate}
The detailed FA is introduced next.

Let ${{\bf{s}}_j} = {[{\Gamma},c_1,c_2,{K_p},{K_i},{K_d},{l_1},{l_2}]^T}$ ($j = 1,...,n$) be the $j$th firefly that encodes all the parameters to be optimized in PMA tracking. The FA minimizes the average error and the maximum deviation:
\begin{align} \label{3.27}
h({{\bf{s}}}_j) = \frac{1}{N}\sum\limits_{t = 1}^N {|{x_d}(t) - x(t)|}   + \lambda  \mathop {\max }\limits_{t \in [1,N]} (|{x_d}(t) - x(t)|)
\end{align}
where $\lambda > 0$ is a trade-off constant, and $N$ the total number of samples. The brightness of the firefly ${{\bf{s}}_j}$, $I_j$, is proportional to $1/h({{\bf{s}}}_j)$.

The attractiveness of the $i$th firefly to the $j$th, $\beta_{ij}$, is:
\begin{align} \label{3.28}
\beta_{ij} = {\beta _0}{e^{ - \gamma {{r_{^{ij}}^2}}}}
\end{align}
where ${\beta _0}$ is the attractiveness at $r = 0$, $\gamma $ is a fixed value related to the length scale in the optimization, and
\begin{align}
{r_{ij}} = \left\| {{{\bf{s}}_i} - {{\bf{s}}_j}} \right\| = \sqrt {{{({{\bf{s}}_i} - {{\bf{s}}_j})}^T}({{\bf{s}}_i} - {{\bf{s}}_j})}.
\end{align}
is the distance between the two fireflies.

The movement of a firefly $i$ to another brighter firefly $j$ is computed by:
\begin{align} \label{3.29}
{{\bf{s}}_{i,k}} = {{\bf{s}}_{i,k}} + {\beta _0}{e^{ - \gamma r_{^{ij}}^2}}({{\bf{s}}_{j,k}} - {{\bf{s}}_{i,k}}) + \alpha (\delta  - \frac{1}{2})
\end{align}
where ${{\bf{s}}_{i,k}}$ is the $k$th element in ${{\bf{s}}_{i}}$, and $\delta  \sim N(0,1)$. The second term comes from the attractiveness, and the third term introduces some randomization.

The optimization is meaningless if the resulting controller is unstable. Thus, special attention is paid to make sure fireflies associated with unstable controllers less attractive. More specifically, when a firefly does not meet the stability conditions, it will not be applied to the plant, and a large penalty term is added to its objective function to make it less attractiveness.

The pseudo-code of the constrained FA is given in Algorithm 1.

\begin{algorithm}[!t]
	\caption{The constrained FA for IDO-PSMC.}
	\begin{algorithmic}[1]
		\STATE Generate initial population of fireflies ${\bf{s}}_i (i=1,2,...,n)$
        \WHILE{$t < MaxGeneration$}
            \FOR{$i=1$ to $n$}
            \IF{${\bf{s}}_i$ satisfies the stability conditions}
                \STATE Apply ${\bf{s}}_i$ to the IDO-PSMC for the PMA.
                \STATE Obtain the experimental results.
            \ENDIF
                \STATE Update objective function $h({\bf{s}}_i)$ and brightness $I_i$.
            \ENDFOR
            \FOR{$i=1$ to $n$}
                \FOR{$j=1$ to $n$}
                    \STATE ${r_{ij}} =  {\left[{{({{\bf{s}}_i} - {{\bf{s}}_j})}^T}({{\bf{s}}_i} - {{\bf{s}}_j})\right]}^{1/2} $
                    \IF{${I_i} < {I_j}$}
                    \STATE Compute the attractiveness $r_{ij}$ in (\ref{3.28}).
                    \STATE Move Firefly $i$ towards Firefly $j$.
                    \ENDIF
                \ENDFOR
            \ENDFOR
        \STATE Rank the fireflies and find the current best
        \ENDWHILE
	\end{algorithmic}
\end{algorithm}

\section{Stability Analysis}

The stability of the proposed controller when $m_p>0$ is analyzed in this section. To our best knowledge, this case has not adequately addressed before.

For the convenience of presentation, we define:
\[{{\bf{K}}_m} = \left[ {\begin{array}{*{20}{c}}
   {{K_i}{c_2}} & 0 & 0  \\
   0 & \varpi  & 0  \\
   0 & 0 & {{K_d}}  \\
\end{array}} \right],\]
where $\varpi  = {K_p}{c_1} - {K_i} - {K_d}{c_2}$.
\begin{theorem}
Consider the nonlinear system in (\ref{2.2}). If Assumption 1 holds, and proper $l_1$ and $l_2$ are chosen to make the real parts of the eigenvalues of ${\bf{A}}_1$ negative, then the norm of tracking error between the proxy states ${\bf{X}}_p$ and the system states ${\bf{X}}$ is uniformly ultimately bounded, and a sliding motion on the surface (\ref{3.14}) can be guaranteed when the IDO-PSMC satisfies:
\[{m_p} > 0,\lambda ({{\bf{K}}_c}) > 0,\Gamma  \ge {\lambda _2}({K_p} + {K_i} + {K_d}),{\varpi } > 0\]
where \[{\lambda _2} = \frac{{(\varepsilon {\rm{ + }}{\lambda _1})({c_1} + {c_2} + 1)}}{{{\lambda _{\min }}({{\bf{K}}_m})}},\]
\[{{\bf{K}}_c} = \left[ {\begin{array}{*{20}{c}}
   {{K_p}{c_2} + {K_i}{c_1}} & {{K_i} + {K_d}{c_2}}  \\
   {{K_i} + {K_d}{c_2}} & {{K_p} + {K_d}{c_1}}  \\
\end{array}} \right].\]
\end{theorem}

\begin{proof}
To simplify the analysis, a Lyapunov candidate is defined as:
\begin{align} \label{4.1}
{V_1} &= \frac{1}{2}{m_p}S_p^2 + \frac{1}{2}S_q^2\\
{V_2} &= \frac{1}{2}({K_p} + {K_d}{c_1} - {K_i} - {K_d}{c_2})\dot e_p^{\rm{2}}\; \nonumber \\
&\quad + \frac{1}{2}({K_p}{c_2} + {K_i}{c_1} - {K_i} - {K_d}{c_2})e_p^{\rm{2}} \nonumber \\
&\quad  + \frac{1}{2}({K_i} + {K_d}{c_2}){({e_p} + {{\dot e}_p})^2} \nonumber \\
& = \frac{1}{2}\left[ {\begin{array}{*{20}{c}}
   {{e_p}} & {{{\dot e}_p}}  \\
\end{array}} \right]{{\bf{K}}_c}\left[ {\begin{array}{*{20}{c}}
   {{e_p}}  \\
   {{{\dot e}_p}}  \\
\end{array}} \right]
\end{align}
where ${e_p} = \displaystyle{\int {\left( {{x_p} - x} \right)dt}}$.

We have $V > 0$, because $\lambda ({{\bf{K}}_c}) > 0$. Its derivatives are:
\begin{align} \label{4.2}
\dot V = \dot V_1 + \dot V_2.
\end{align}
From (\ref{3.16})-(\ref{3.21}), let
\begin{align} \label{4.3}
{m_p}{{\dot S}_p}& =  - {\Gamma}{\mathop{\rm sgn}} ({S_p}) + {K_p}({x_p} - x) \nonumber \\
&\quad + {K_i}\displaystyle{\int {\left( {{x_p} - x} \right)dt}} + {K_d}({{\dot x}_p} - \dot x) \nonumber \\
&  =  - \Gamma {\rm{sgn}}({S_p}) + {K_p}{{\dot e}_p} + {K_i}{e_p} + {K_d}{{\ddot e}_p}.
\end{align}
According to (\ref{3.15}), we have
\begin{align} \label{4.4}
{{\dot S}_q} &=  - {K_p}({x_p} - x) - {K_i}\displaystyle{\int {\left( {{x_p} - x} \right)dt} } \nonumber \\
&\quad { - {K_d}({{\dot x}_p} - \dot x) - \tilde \tau  + \hat {\dot \tau} } \nonumber \\
&  =  - {K_p}{{\dot e}_p} - {K_i}{e_p} - {K_d}{{\ddot e}_p} - \tilde \tau  + \hat {\dot \tau }.
\end{align}
Also, combining (\ref{3.13}) and (\ref{3.14}), it follows that
\begin{align} \label{4.5-3}
{S_p} = {S_q} - ({\ddot e_p} + {c_1}{\dot e_p} + {c_2}{e_p}).
\end{align}
Integrating (\ref{4.2})-(\ref{4.5-3}), the derivatives of $V_1$ and $V_2$ are:
\begin{align}
{\dot V}_1 &= {S_p}( - \Gamma {\rm{sgn}}({S_p}) + {K_p}{{\dot e}_p} + {K_i}{e_p} + {K_d}{{\ddot e}_p})\nonumber\\
&\quad  + {S_q}( - {K_p}{{\dot e}_p} - {K_i}{e_p} - {K_d}{{\ddot e}_p} - \tilde \tau  + \hat {\dot \tau} ) \nonumber\\
 &=  - \Gamma |{S_p}| + (\dot \tau  - \tilde \tau  - \tilde {\dot \tau} ){S_q}\nonumber\\
&\quad + ({K_p}{{\dot e}_p} + {K_i}{e_p} + {K_d}{{\ddot e}_p})({S_p} - {S_q})\nonumber\\
 &=  - \Gamma |{S_p}| + (\dot \tau  - \tilde \tau  - \tilde {\dot \tau} ){S_q}\nonumber\\
&\quad  + ({K_p}{{\dot e}_p} + {K_i}{e_p} + {K_d}{{\ddot e}_p})( - {{\ddot e}_p} - c_1{{\dot e}_p} - c_2 {e_p})\nonumber\\
&= - \Gamma |{S_p}| + (\dot \tau  - \tilde \tau  - \tilde {\dot \tau} ){S_q} - {K_d}\ddot e_p^2 - {K_p}{c_1}\dot e_p^2 \nonumber\\
&\quad - {K_i}{c_2}e_p^2 - ({K_p} + {K_d}{c_1}){{\dot e}_p}{{\ddot e}_p} \nonumber\\
&\quad - ({K_i} + {K_d}{c_2}){e_p}{{\ddot e}_p} - ({K_p}{c_2} + {K_i}{c_1}){e_p}{{\dot e}_p}.\label{4.5}\\
 {{\dot V}_2}& = ({K_p} + {K_d}{c_1} - {K_i} - {K_d}{c_2}){{\dot e}_p}{{\ddot e}_p} \nonumber\\
&\quad  + ({K_p}{c_2} + {K_i}{c_1} - {K_i} - {K_d}{c_2}){e_p}{{\dot e}_p} \nonumber\\
&\quad  + ({K_i} + {K_d}{c_2})({e_p}{{\dot e}_p} + {e_p}{{\ddot e}_p} + \dot e_p^{\rm{2}}\; + {{\dot e}_p}{{\ddot e}_p}) \nonumber\\
& = ({K_p} + {K_d}{c_1}){{\dot e}_p}{{\ddot e}_p} + ({K_p}{c_2} + {K_i}{c_1}){e_p}{{\dot e}_p} \nonumber\\
&\quad  + ({K_i} + {K_d}{c_2})\dot e_p^{\rm{2}}\; + ({K_i} + {K_d}{c_2}){e_p}{{\ddot e}_p}. \label{4.5-1}
\end{align}
Then, it follows that
\begin{align}
 \dot V& = {{\dot V}_1} + {{\dot V}_2} \nonumber\\
 & =  - \Gamma |{S_p}| + (\dot \tau  - \tilde \tau  - \tilde {\dot \tau} ){S_q} - {K_d}\ddot e_p^2 \nonumber\\
&\quad  - ({K_p}{c_1} - {K_i} - {K_d}{c_2})\dot e_p^2 - {K_i}{c_2}e_p^2 \nonumber\\
 & = - \Gamma |{S_p}| + (\dot \tau  - \tilde \tau  - \tilde {\dot \tau} ){S_q} - {K_d}\ddot e_p^2 \nonumber\\
&\quad  - \varpi \dot e_p^2 - {K_i}{c_2}e_p^2.\label{4.5-2}
\end{align}

Because
\begin{align}
 \Gamma  &\ge {\lambda _2}({K_p} + {K_i} + {K_d}) \nonumber \\
&  = \frac{{({K_p} + {K_i} + {K_d})({c_1} + {c_2} + 1)}}{{\min \{ {K_i}{c_2},\varpi ,{K_d}\} }}(\varepsilon  + {\lambda _1}) \nonumber \\
& \ge (\varepsilon  + {\lambda _1}).
\end{align}
it follows that
\begin{align}
 \dot V &\le  - \Gamma \left| {{S_q} - ({{\ddot e}_p} + {c_1}{{\dot e}_p} + {c_2}{e_p})} \right| \nonumber\\
&\quad  + (\varepsilon  + \left\| {{\bf{\tilde e}}} \right\|)\left| {{S_q}} \right| - {K_d}\ddot e_p^2 - \varpi \dot e_p^2 - {K_i}{c_2}e_p^2\nonumber \\
&  \le  - (\varepsilon  + {\lambda _1})\left| {{S_q} - ({{\ddot e}_p} + {c_1}{{\dot e}_p} + {c_2}{e_p})} \right| \nonumber\\
&\quad  + (\varepsilon  + \left\| {{\bf{\tilde e}}} \right\|)\left| {{S_q}} \right| - {K_d}\ddot e_p^2 - \varpi \dot e_p^2 - {K_i}{c_2}e_p^2 \nonumber\\
&  \le  - (\varepsilon  + {\lambda _1})\left| {{S_q}} \right| + (\varepsilon  + {\lambda _1})\left| {{{\ddot e}_p} + {c_1}{{\dot e}_p} + {c_2}{e_p}} \right| \nonumber\\
&\quad  + (\varepsilon  + {\lambda _1})\left| {{S_q}} \right| - {K_d}\ddot e_p^2 - \varpi \dot e_p^2 - {K_i}{c_2}e_p^2 \nonumber\\
&  \le (\varepsilon  + {\lambda _1})({c_1} + {c_2} + 1)\left\| {{{\bf{e}}_p}} \right\| - {\lambda _{\min }}({{\bf{K}}_m}){\left\| {{{\bf{e}}_p}} \right\|^2} \nonumber\\
&  =  - \left\| {{{\bf{e}}_p}} \right\|[{\lambda _{\min }}({{\bf{K}}_m})\left\| {{{\bf{e}}_p}} \right\| - (\varepsilon  + {\lambda _1})({c_1} + {c_2} + 1)] \label{4.6}
\end{align}
where ${{\bf{e}}_p} = {{\bf{X}}_p} - {\bf{X}} = {\left[ {\begin{array}{*{20}{c}}
   {{e_p}} & {{{\dot e}_p}} & {{{\ddot e}_p}}  \\
\end{array}} \right]^T}$.
It is easy to see that after a sufficiently long time
\begin{align} \label{e_norm}
\left\| {{{\bf{e}}_p}} \right\| \le {\lambda _2},
\end{align}
where \[{\lambda _2} = \frac{{(\varepsilon  + {\lambda _1})({c_1} + {c_2} + 1)}}{{{\lambda _{\min }}({{\bf{K}}_m})}}.\]
Therefore, $\left\| {{{\bf{e}}_p}} \right\|$ is uniformly ultimately bounded.

Define a new Lyapunov candidate as:
\begin{align}
{V_3} = S_p^2.
\end{align}
It follows from (\ref{4.3}) that
\begin{align}
 {{\dot V}_3} &= {S_p}{{\dot S}_p} \nonumber\\
& =\frac{{{S_p}}}{{{m_p}}}( - \Gamma {\mathop{\rm sgn}} ({S_p}) + {K_p}{{\dot e}_p} + {K_i}{e_p} + {K_d}{{\ddot e}_p}) \nonumber\\
&  \le  - \frac{\Gamma }{{{m_p}}}|{S_p}| + \frac{{\lambda _2}}{{{m_p}}}{({K_p} + {K_i} + {K_d})}{S_p} \nonumber\\
& \le 0.
 \end{align}
When $||{\bf{e}}_p||$ is uniformly ultimately bounded, the achievement of a sliding motion on the surface (\ref{3.14}) is guaranteed. This completes the proof.
\end{proof}

\begin{remark}
The stability analysis of the system has two steps. First, the norm of the tracking error between the proxy states ${\bf{X}}_p$ and the system states ${\bf{X}}$ is uniformly ultimately bounded, which indicates that the system states converge to the proxy states. Then, the achievement of a sliding motion on the surface (\ref{3.14}) means that the proxy can track the reference. In summary, the system states approach the reference, and the stability of the closed-loop system is guaranteed.
\end{remark}

\begin{corollary}
If Inequality (\ref{e_norm}) holds, and initially $x_p = x_d$, then, as the proxy mass $m_p$ increases, the upper bound of $S_q$ will gradually approach a bound associated with the system's disturbances.
\begin{align} \label{4.14}
\mathop {\lim }\limits_{{m_p} \to \infty } |{S_q}| \le {\lambda _2}(c_1 + c_2 + 1)
\end{align}
\end{corollary}

\begin{proof}
From (\ref{4.3}), it follows that
\begin{align} \label{4.8}
|{\dot S_p}| = \frac{1}{{{m_p}}}| - \Gamma {\mathop{\rm sgn}} ({S_p}) + {K_p}{\dot e_p} + {K_i}{e_p} + {K_d}{\ddot e_p}|
\end{align}
Since the system is globally uniformly ultimately bounded, we have
\begin{align} \label{4.9}
\mathop {\lim }\limits_{{m_p} \to \infty } |{{\dot S}_p}| = 0.
\end{align}
The proxy mass $m_p$ is a fixed value in each experiment. Let $t_f$ be the finite duration of the experiment. Then,
\begin{align} \label{4.10}
{S_p} = \int_0^{{t_f}} {{{\dot S}_p}} dt + \varpi
\end{align}
where $\varpi $ is the initial value of $x_d - x_p$, which equals zero. Hence, it follows that
\begin{align} \label{4.11}
|{S_p}| = |\int_0^{{t_f}} {{{\dot S}_p}} dt| \le \int_0^{{t_f}} {|{{\dot S}_p}} |dt.
\end{align}
Combining (\ref{4.9}) and (\ref{4.11}), we can obtain
\begin{align} \label{4.12}
\mathop {\lim }\limits_{{m_p} \to \infty } |{S_p}| = 0.
\end{align}
Considering (\ref{4.5-3}) and (\ref{e_norm}), after a sufficiently long time
\begin{align}
 |{S_q}| &\le |{S_p}| + |{{\ddot e}_p} + {c_1}{{\dot e}_p} + {c_2}{{\ddot e}_p}| \nonumber \\
&  \le |{S_p}| + {\lambda _2}(c_1 + c_2 + 1). \label{4.13}
\end{align}
Finally,
\begin{align} \label{4.14}
\mathop {\lim }\limits_{{m_p} \to \infty } |{S_q}| \le {\lambda _2}(c_1 + c_2 + 1).
\end{align}
According to the above results, when the proxy mass $m_p$ approaches positive infinity, the upper bound of $S_q$ will approach a bound associated with the disturbances of the system.
\end{proof}

Therefore, $m_p$ can be used to adjust the tracking speed of the proxy. Normally, it should be sufficiently large, so that the proxy trajectory will track the reference accurately, due to $|{S_p}| \to 0$. In this case, the IDO-PSMC is like a stable PID-based algorithm, in which $S_q$ is limited to a relatively small bound.

\section{Experiments}

Experiments are performed in this section to demonstrate the performance of the proposed IDO-PSMC.

\subsection{Experiment Setup}

The core part of the physical hardware system is an xPC target from MathWorks. It enables a host computer  installed with MATLAB/SIMULINK, C compiler and applied models to generate executable codes, while a target executes the generated code in real-time.

The PMA in the experiments, shown in Fig.~\ref{physical_platform}, was Festo DMSP-20-200N-RM-RM fluidic muscle with an internal diameter of 20 mm, nominal length of 200 mm, and an operating pressure range from 0 to 6 bar. The Festo VPPM-6L-L-1-G18-0L10H-V1P proportional valve was used to regulate the pressure of the compressed air inside the PMA. The displacement sensor was GA-75, whose measurement range was 0-150 mm.

\begin{figure}[htpb] \centering
\includegraphics[width=3.2in,height=2.5in]{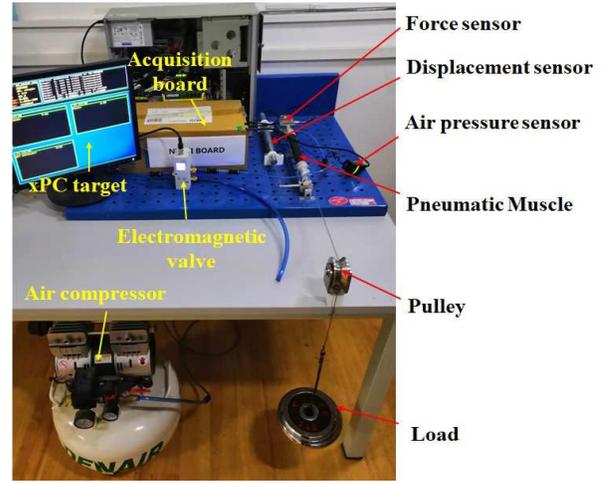}
\caption{The PMA system.}
\label{physical_platform}
\end{figure}

The proposed method does not require an accurate three-element model of the PMA. So, we used the identified parameters of a similar PMA in \cite{Zhu:2017} (see Table I). The damping coefficient $b(P)$ is dependent on whether the PMA is being inflated or deflated, which corresponds to two sets of $b_0$ and $b_1$. The spring coefficient $k(P)$ is a piecewise linear function at point $P = 325420$ Pa. Hence there are also two different sets of $k_0$, $k_1$ for $P < 325420$ Pa and $P > 325420$ Pa, respectively.

We designed two reference trajectories. The first was a fixed frequency sinusoid:
\begin{align}
{x_d} = {A_x}\sin (2\pi f_xt) + {B_x}
\end{align}
where $A_x = 0.015$ m, $f_x = 0.25$ Hz, and $B_x = 0.015$ m. The second was a sine wave whose frequency changed linearly from $0.1$ Hz to $0.5$ Hz within 20 s. The sampling time was set to 0.001 s.

\begin{table}[htpb] \centering
\caption{The model parameters.}
\label{table_model_parameters}
\begin{tabular}{c|c|c|c}
\hline
\hline
Parameter & Value (Unit) & Parameter & Value (Unit) \\
\hline
$f_0 $          & $-202.32$ (N)          &    $f_1 $         & $0.00721$ (N/Pa)        \\
$k_{01} $       & $18063.0$ (N/m)        &    $k_{02} $    & $0.01051$ (N/(m.Pa))      \\
$k_{11} $       & $-0.2132$ (N/m)        &    $k_{12} $    & $90638.0$ (N/(m.Pa))       \\
$b_{0i} $     & $6435.31$ (N.s/m)      &   $b_{1i} $    & $0.10023$ (N.s/(m.Pa))      \\
$b_{0d} $     & $2522.01$ (N.s/m)      &    $b_{1d} $    & $0.00321$ (m.s)       \\
\hline
\hline
\end{tabular}
\end{table}

The maximum absolute error (MAE) and the integral of absolute error (IAE) were used as our performance measures:
\begin{align}
MAE{R^a} &= Max(|{x_d}(t) - x(t)|_{t = 1}^N) \\
IAE{R^b} &= \frac{1}{N}\sum\limits_{t = 1}^N {|{x_d}(t) - x(t)|}
\end{align}

\subsection{Experimental Results}

Several experiments were performed to verify the effectiveness of the proposed control strategy in handling the disturbances/uncertainties. First, experiments were conducted to verify Corollary~1 based on different $m_p$ values. Then, both the fixed frequency sinusoid and varying frequency sinusoid were used as the reference trajectories to show the superiority of the proposed method, compared with PSMC, SMC, and DO-SMC. Finally, experiments of the PMA with different loads (0, 2.5, 5 kg) were performed to demonstrate the robustness of the IDO-PSMC. The parameters of all the control strategies, IDO-PSMC, PSMC, DO-SMC, and SMC, were optimized by the FA. The following set of control parameters of the IDO-PSMC were used in all experiments: ${\Gamma} = 14218.8$, $c_1 = 177.4$, $c_2 = 174.4$, ${K_p}=2473.5$, ${K_i}=1916$, ${K_d}=194.2$, $l_1=15952$, and $l_2=0$.

Fig.~\ref{exp_m_p} shows the experimental results when $m_p=\{0.5, 1.0, 5.0, 10.0, 15.0\}$. As $m_p$ increased, the tracking accuracy improved, and the variation of $S_q$ also significantly decreased. The control performance saturated when $m_p$ was sufficient large. The MAEs and IAEs are shown in Table \ref{table_exp_m_p}.

\begin{figure}[htpb] \centering
\includegraphics[width=3.4in,height=2.125in]{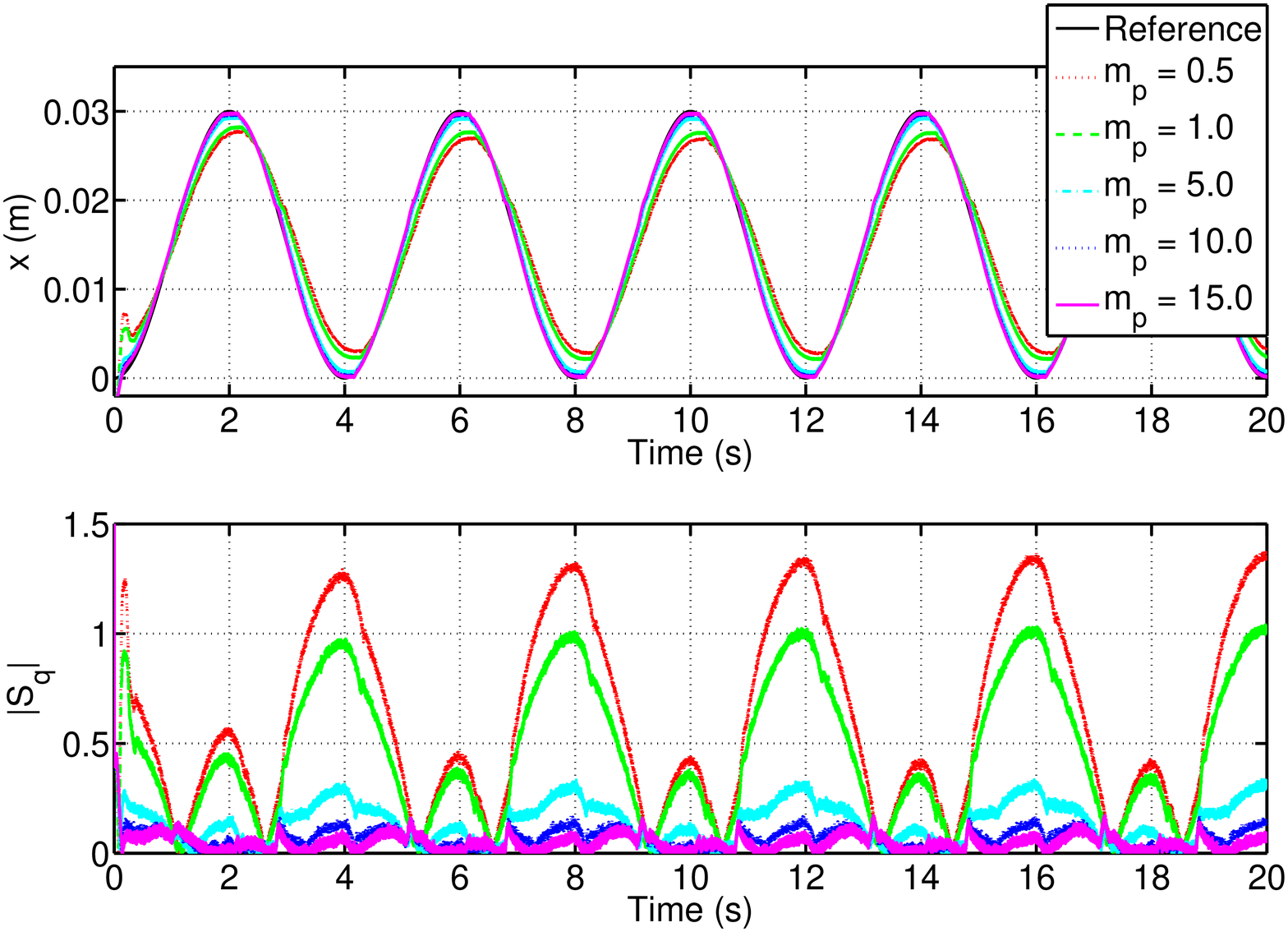}
\caption{Tracking performance of the IDO-PSMC with different $m_p$ values.}
\label{exp_m_p}
\end{figure}

\begin{table}[htpb] \centering
\caption{Tracking performance of the IDO-PSMC with different $m_p$ values.}
\label{table_exp_m_p}
\begin{tabular}{l|c|c}
\hline
\hline
              & MAE   & IAE    \\
\hline
$m_p = 0.5$      &  4.5 $\times {10^{-3}}$ (m)     & 2.6 $\times {10^{-3}}$ (m)                          \\

$m_p = 1.0$          & 3.4 $\times {10^{-3}}$ (m)   &  2.0 $\times {10^{-3}}$ (m)                              \\

$m_p = 5.0$        & 1.1 $\times {10^{-3}}$ (m)              &     5.3 $\times {10^{-4}}$ (m)                         \\

$m_p = 10.0$          & 7.1 $\times {10^{-4}}$ (m)              &       2.1 $\times {10^{-4}}$ (m)                            \\

$m_p = 15.0$          & 5.5 $\times {10^{-4}}$ (m)              &       1.6 $\times {10^{-4}}$ (m)                            \\
\hline
\hline
\end{tabular}
\end{table}

Fig.~\ref{exp_fix_freq_no_load} shows the performance of different control strategies with the fixed-frequency ($0.25$ Hz) sinusoidal reference. We replaced the $sgn$ function of the SMC and DO-SMC with a $sat$ function to alleviate chattering, but the SMC and DO-SMC still resulted in chattering, due to the inaccurate model parameters in Table~\ref{table_model_parameters}. The basic PSMC enabled the PMA to track the reference with acceptable precision, since it is a model-free strategy, not affected by inaccurate model parameters. However, its performance was still much worse than the IDO-PSMC. The corresponding MAEs and IAEs of all four control strategies are shown in Table~\ref{table_exp_tracking}. Since we were more interested in the steady state performance, the values were calculated from 2 s to 20 s.

\begin{figure}[htpb] \centering
\includegraphics[width=3.4in,height=2.125in]{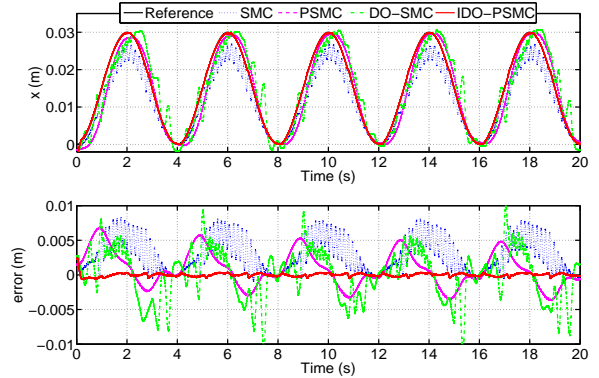}
\caption{Tracking performance of different control strategies with the fixed-frequency sinusoidal reference ($0.25$ Hz).}
\label{exp_fix_freq_no_load}
\end{figure}

\begin{table}[htpb] \centering
\caption{Tracking performance of different control strategies with the fixed-frequency sinusoidal reference ($0.25$ Hz).}
\label{table_exp_tracking}
\begin{tabular}{l|c|c}
\hline
\hline
              & MAE  & IAE   \\
\hline
IDO-PSMC      & 5.5 $\times {10^{-4}}$ (m)   & 1.6 $\times {10^{-4}}$ (m)                             \\

PSMC          & 5.8 $\times {10^{-3}}$ (m)   & 1.9 $\times {10^{-3}}$ (m)                               \\

DO-SMC        & 1.1 $\times {10^{-2}}$ (m)               & 3.0 $\times {10^{-3}}$ (m)                            \\

SMC           & 8.1 $\times {10^{-3}}$ (m)               & 2.8 $\times {10^{-3}}$ (m)                                 \\
\hline
\hline
\end{tabular}
\end{table}

Fig.~\ref{exp_var_freq_no_load} shows the tracking performance of different control strategies with the varying-frequency (0.1-0.5Hz) sinusoidal reference. Again, the proposed IDO-PSMC performed the best among all four control strategies, although it had a large oscillation at the beginning. This is because that it needs some time to drive the states of the PMA into the boundary [see (\ref{e_norm})]. After that, the proposed IDO-PSMC can handle the system disturbances/uncertainties, and achieve accurate tracking. The corresponding MAEs and IAEs of all four control strategies are shown in Table~\ref{table_exp_freq_tracking}. Again, the values were calculated from 2 s to 20 s.

\begin{figure}[htpb] \centering
\includegraphics[width=3.4in,height=2.125in]{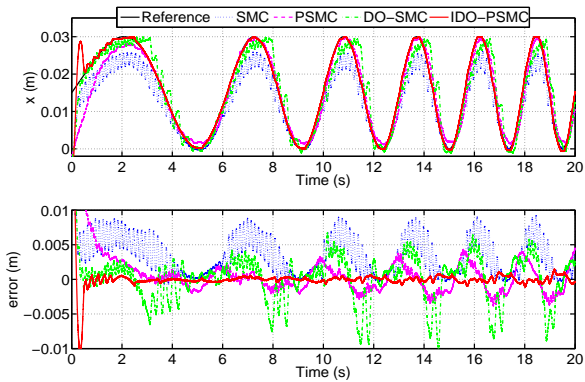}
\caption{Tracking performance of different control strategies with the varying-frequency sinusoidal reference ($0.1-0.5$ Hz).} \label{exp_var_freq_no_load}
\end{figure}

\begin{table}[htpb] \centering
\caption{Tracking performance of different control strategies with the varying-frequency sinusoidal reference ($0.1-0.5$ Hz).} \label{table_exp_freq_tracking}
\begin{tabular}{l|c|c}
\hline
\hline
              & MAE   & IAE   \\
\hline
IDO-PSMC      & 1.5 $\times {10^{-3}}$ (m)   & 2.9 $\times {10^{-4}}$ (m)                             \\

PSMC          & 4.5 $\times {10^{-3}}$ (m)   & 1.4 $\times {10^{-3}}$ (m)                               \\

DO-SMC        & 1.0 $\times {10^{-2}}$ (m)               & 2.6 $\times {10^{-3}}$ (m)                            \\

SMC           & 9.3 $\times {10^{-3}}$ (m)               & 3.0 $\times {10^{-3}}$ (m)                                 \\
\hline
\hline
\end{tabular}
\end{table}

To further investigate the robustness of the proposed control strategy, different loads were attached to the PMA, while tracking the varying-frequency (0.1-0.5 Hz) reference. The IDO-PSMC results are shown in Fig.~\ref{exp_var_freq_load} and Table~\ref{table_exp_load_freq_tracking}. Generally they were very robust to changing loads. However, a closer-look reveals that as the load increased, the tracking performance slightly deteriorates. This is because the fixed parameters of the DO can only handle a certain amount of disturbances. When the disturbance is too much, the parameters of DO have to be re-tuned.

\begin{figure}[htpb] \centering
\includegraphics[width=3.4in,height=2.125in]{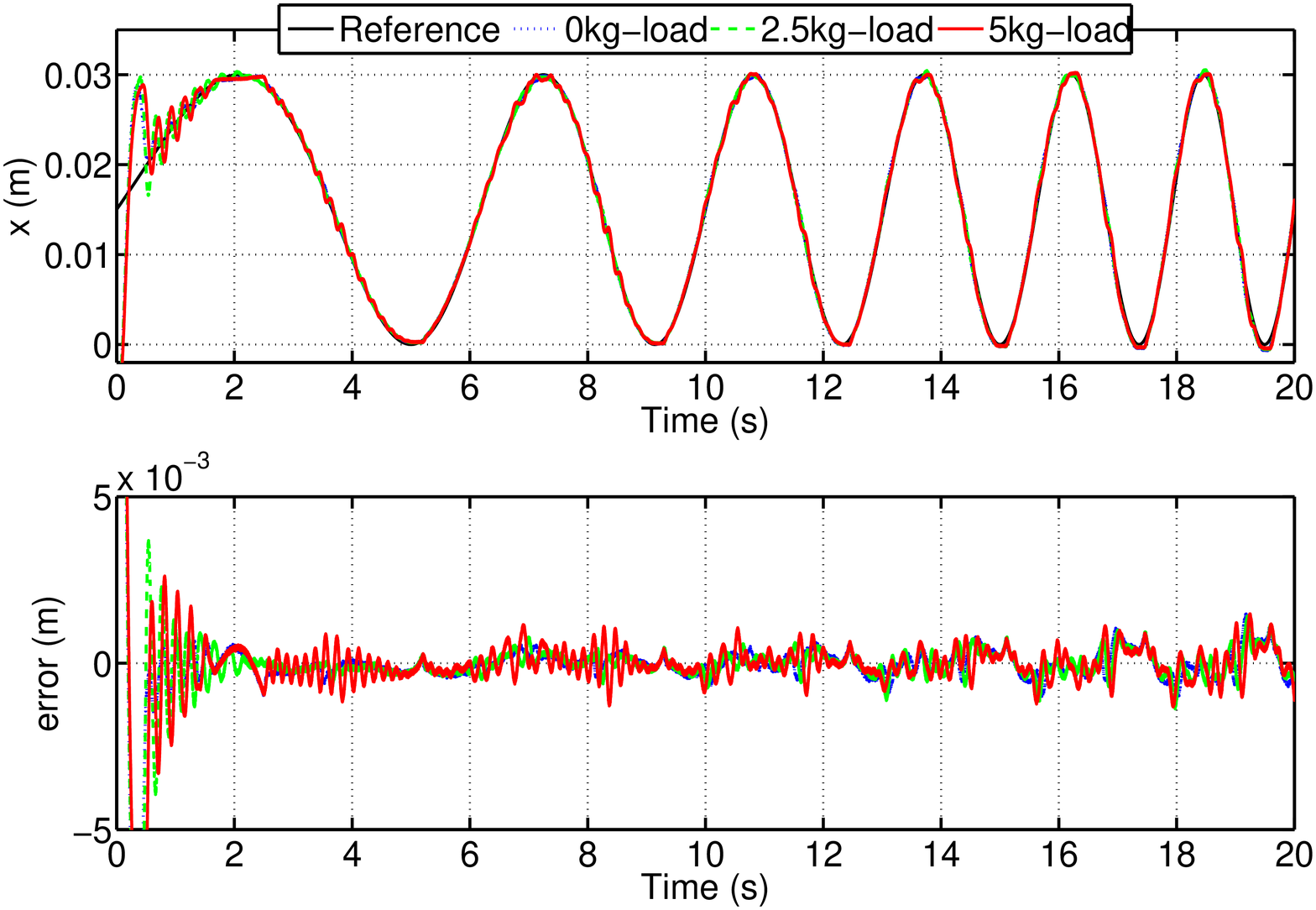}
\caption{Tracking performance of the IDO-PSMC with the PMA attaching different loads.}
\label{exp_var_freq_load}
\end{figure}

\begin{table}[htpb] \centering
\caption{Tracking performance of the IDO-PSMC with the PMA attaching different loads.}
\label{table_exp_load_freq_tracking}
\begin{tabular}{l|c|c}
\hline
\hline
              & MAE   & IAE    \\
\hline
0kg-load      & 1.5 $\times {10^{-3}}$ (m)   & 2.9 $\times {10^{-4}}$ (m)                             \\

2.5kg-load          & 1.4 $\times {10^{-3}}$ (m)   & 2.6 $\times {10^{-4}}$ (m)                               \\

5kg-load        & 1.5 $\times {10^{-3}}$ (m)               & 3.3 $\times {10^{-4}}$ (m)                            \\
\hline
\hline
\end{tabular}
\end{table}

\section{Conclusion}

This paper presented a robust control strategy, IDO-PSMC, for the PMA. The DO was used to deal with the estimated disturbance. The tracking states of the PMA were proven to be uniformly ultimately bounded according to the Lyapunov theorem. A constrained FA was used to tune the controller parameters automatically. Extensive experiments were conducted to demonstrate the superior performance of the proposed IDO-PSMC. Compared with other control strategies, IDO-PSMC can adequately handle the uncertainties/disturbances, and achieve better control performance in high-precision tracking. To the best of our knowledge, this is the first attempt to control the PMA with a varying-frequency reference trajectory and different loads.

\appendices

\ifCLASSOPTIONcaptionsoff
  \newpage
\fi

\end{document}